\documentclass[11pt,a4paper]{article}

\usepackage{xcolor}
\title{Design of MDP Convolutional Codes and Maximally Recoverable Codes Through the Lens of Matrix Completion}

\usepackage{authblk}
\usepackage{amsmath}
\usepackage{amssymb}
\usepackage{xcolor}
\usepackage{amsthm}
\usepackage{comment}
\usepackage{mathdots}
\usepackage[hidelinks]{hyperref}

\newtheorem{theorem}{Theorem}

\newtheorem{lemma}[theorem]{Lemma}

\newtheorem{remark}{Remark}
\newtheorem{definition}{Definition}
\newtheorem{corollary}{Corollary}

\author[1]{Sakshi Dang}
\author[2]{Julia Lieb}
\author[3]{Pedro Soto}
\author[4]{Alex Sprintson}
\affil[1]{Indian Institute of Science, Bengaluru, India\thanks{sakshidang10@gmail.com
}}
\affil[2]{TU Ilmenau, Germany\thanks{julia.lieb@tu-ilmenau.de}}
\affil[3]{Virginia Tech,  USA\thanks{pedro.juan.soto.conde@gmail.com}}
\affil[4]{George Mason University, USA\thanks{asprints@gmu.edu}}

\date{}
\usepackage{amsmath}
\usepackage{amssymb}
\begin{document}

\maketitle

\begin{abstract}
The matrix completion problem provides a unifying lens through which many fundamental problems in coding theory can be viewed. 
In this paper, we investigate  Locally Recoverable Codes (LRCs)  with Maximal Recoverability (MR) 
and Maximum Distance Profile (MDP) convolutional codes in the framework of matrix completion.
 In particular, we present techniques that are general enough to provide constructions for both types of codes. 
  A common feature of our code constructions is the sparsity of their generator matrices and the property that a large number of the entries of the generator matrices are elements of a small subfield of a larger extension field. 
\end{abstract}

\medskip


\medskip

\section{Introduction}

The \emph{matrix completion problem} provides a powerful and unifying perspective for understanding a broad class of fundamental problems in coding theory and related areas. In particular, many code design problems, including the construction of convolutional codes, distributed storage codes, and codes for wireless communication, can be naturally formulated as variants of the matrix completion problem. 
In its most general formulation, the problem considers a partially specified matrix subject to structural constraints, such as 
column dependencies, and block-structure requirements, together with unconstrained (free) entries. 
These constraints can be captured algebraically by specifying a set of polynomial equations over the matrix entries and fixing selected entries to given values.
The goal of the matrix completion problem is to assign values to the free entries to satisfy specified algebraic objectives, such as 
ensuring the non-vanishing of specified full-size minors, while minimizing the required finite field size. The overarching challenge is to develop principled matrix completion frameworks that apply uniformly across a wide range of constraint patterns.


One class of matrix completion problems concerns the construction of a generator matrix of an MDS code that has a prescribed zero pattern (support constraint). The related GM-MDS\footnote{GM-MDS conjecture stands for Generator Matrix - Maximum Distance Separable conjecture.} conjecture, introduced
by Dau, Song, and Yuen 
in \cite{DSY14}, posits that the natural combinatorial feasibility criterion 
often called the
\emph{MDS condition} is not only necessary, but also sufficient to guarantee the existence of an MDS generator matrix over any field of size linear in the length of the code. The conjecture was subsequently resolved in full in \cite{Lovett21,YH19}.


Also the design of \emph{Maximum Distance Profile} (MDP) convolutional codes (commonly referred to as MDP codes) can be viewed as a variation of the matrix completion problem \cite{itw}. Convolutional codes are important for low-delay encoding and decoding of a stream of data and MDP codes are the class of convolutional codes that allows for best possible error-correction in this setting \cite{tomas2012decoding}. The sliding generator matrix of a convolutional code takes a specific form in which multiple copies of blocks are arranged in a repeating sliding pattern and some entries are fixed to zero. While this structural constraint 
does not allow every full-size minor to be non-zero, sliding generator matrices of MDP codes have the property that each
full size minor that can be nonzero, given the prescribed structural requirement, is nonzero,
\cite{MDS-Conv}. 
There exist several constructions for MDP codes, however mostly requiring large finite fields; see e.g., \cite{alfarano2024weighted,chen2023lower,cheng2026new,luo2023construction, almeida2013new}.


Another problem that can be viewed through the lens of Matrix Completion is the design of \emph{Locally Recoverable Codes} (LRCs) that are  \emph{Maximally Recoverable} (MR). The local recoverability property implies the existence of a collection of \emph{repair groups} such that each group enables local recovery of its symbols; that is,  symbol(s) in a repair group can be recovered using only the symbols within that group.
\mbox{MR-LRCs} provide recoverability from the largest possible set of erasure patterns (see \cite{BHH13,GHJY14,GG22,CMST22,umberto} and references therein) among linear block codes fulfilling certain locality constraints. In this paper, we focus on a class of \mbox{MR-LRCs} with unequal locality, i.e., codes in which the subcodes that correspond to different repair groups might have different sizes and dimensions, in addition to a set of global parity symbols which are not included in the repair groups.



We present a matrix-completion technique that is broad enough to encompass the construction of both code families. A unifying aspect of these constructions is a sparse generator matrix with elegant structural properties 
including that certain submatrices of the generator matrix are
constructed using a small subfield of a larger extension field. Such a structure facilitates efficient encoding and decoding procedures. 
In Section 2, we provide preliminaries on matrix completion, MR-LRCs and MDP codes. In Section 3, we present a construction for MR-LRCs with unequal localities for certain parameters. In Section 4, we present constructions for MDP codes with certain parameters. In Section 5, we explain commonalities of these constructions and mention some related matrix completion problems to be considered for future research.


\section{Preliminaries}

We start with giving the matrix completion problem in a general form before considering specific instances of it, i.e. the construction of certain MR-LRCs and MDP codes.

\begin{definition}\label{def:mrc}
    Let $M$ be a $K\times N$ matrix whose entries $m_{i,j}$ for $i=1,\hdots,K$, $j=1,\hdots,N$ are considered as variables. Define $\mathbb F_q[M]:=\mathbb F_q[m_{1,1},\ldots,m_{K,N}]$ and consider a
    set of polynomials $f_\ell(M)\in\mathbb F_q[M]$ for $\ell=1,\hdots,s$. 
    We say that a \textbf{solution} to the \textbf{matrix completion problem} corresponding to 
    $f_1,\hdots,f_s$ is a set of values for the $m_{i,j}\in\mathbb F_q$ that satisfies the following conditions: 
\begin{enumerate}
    \item $f_{\ell}(M) = 0$ for $\ell=1,\hdots,s$ (these equations are called \textbf{constraints} of the matrix completion problem);
    \item all \textbf{non-trivial full-size minors}, \emph{i.e.,} the $K \times K$ minors of $M$ that are not nilpotent in the ring $\mathbb{F}_q [M ] / \left\langle f_1(M),\hdots,f_s(M) \right\rangle$, are non-zero at the values chosen for $m_{i,j}$ ($f$ is nilpotent if there exists a positive integer $n$ such that $f^n = 0$).
\end{enumerate}

\end{definition}

\subsection{Maximally Recoverable Locally Recoverable Codes with Unequal Locality (MR-LRCs with UL)}

In this paper, we focus on Locally Recoverable Codes with Unequal Locality (LRCs with UL) \cite{7541336, 7541377}, defined as follows.  

\begin{definition}
     Let $\mathcal{C}$ be a non-degenerate linear block code of length $L$.
     We say that a codeword symbol $c_j$ for $j\in\{1,\hdots,L\}$ has \textbf{locality} $k_j$ if there exists   a \textbf{repair group} $R_j \subset \{1,\hdots,L\}$ such that $ j \in R_j$ and $\mathcal{C}|_{R_j}$ forms a $[|R_j|, k_j, \delta_j]$-code; if $\delta_j>1$ for all $j$, then we call it a \textbf{Locally Recoverable Code with Unequal Locality}.
     %
     \end{definition}
 
There exists a modified Singleton bound on the minimum distance of LRCs with UL \cite{7541336, 7541377}. 
We extend LRCs with UL to MR-LRCs with UL by requiring maximum recoverability, i.e. all erasure patterns that can be recovered under the given locality constraints should be recovered. This property is considerably stronger than reaching the modified Singleton bound and
captures the code's full erasure-correction potential. 

\begin{definition}\label{def:1}
For $\ell,h,n_i,k_i\in\mathbb N$ with $n_i\geq k_i$ for $i=0,\hdots,\ell-1$, let $\mathcal{C}\subseteq \mathbb{F}_{q}^{N}$ be an $\mathbb{F}_{q}$-linear code of length $N=n_0+\cdots+n_{\ell-1}+h$ and dimension $K=k_0+\cdots+k_{\ell-1}$. 
Set $\mathcal{G}_i:=\{1+\sum_{t=0}^{i-1}n_t,\hdots,\sum_{t=0}^i n_t\}$ for $i=0,\hdots,\ell-1$. Note that these sets form a partition of $\{1,\hdots,N-h\}$ and $|\mathcal{G}_i|=n_i$.
We say that $\mathcal{C}$ is \textbf{MR-LRC with UL} if 
   \begin{enumerate}
         \item $\dim(\mathcal{C}|_{\mathcal{G}_i})\leq k_i$, i.e.\ each codeword symbol indexed by $\mathcal{G}_i$ has locality $k_i$;
        \item 
        for each \textbf{admissible} puncturing $S$, i.e. for each puncturing where 
        exactly $n_i-k_i$ coordinates in $\mathcal{G}_i$ are deleted, 
        the punctured code
$\mathcal{C}|_S 
$ 
        %
        is an $[h+K ,K]$-MDS code.
    \end{enumerate}
   The sets $\mathcal{G}_i$ are called  \textbf{local repair groups} and the last $h$ coordinates of a codeword are called the \textbf{global parities}.
\end{definition}

\begin{remark}
Note that 
Definition~\ref{def:1} implies that $\mathcal{C}|_{\mathcal{G}_i}$ is an $[n_i,k_i]$-MDS code for $i=0,\ldots,\ell-1$.
Also note that in contrast to the MR-LRCs considered in the literature before (see e.g. \cite{GG22}), by considering MR-LRCs with UL we allow the local repair groups to have different sizes $n_i$ and localities $k_i$ (while the $h$ global parities have locality $K$) but require that $k_0+\cdots+k_{\ell-1}$ equals the dimension $K$ of the full code. MR-LRCs are also known as partial MDS codes \cite{BHH13}. In \cite{aa20}, partial MDS codes with different sizes $n_i$ are considered but still it is required that all localities $k_i$ are equal.
\end{remark}

\begin{remark}\label{compmr}
    The problem of constructing MR-LRCs with UL can be seen as a matrix completion problem where Condition 1 of Definition \ref{def:1} provides the constraints and Condition 2 describes the non-trivial full-size minors of a generator matrix $M$ of $\mathcal{C}$.
    More precisely, if
   $M_i$ is the $K \times n_i$ submatrix of $M$ given by the columns indexed by $\mathcal{G}_i$, Condition 1 of Definition~\ref{def:1} 
   is equivalent to the 
   constraints
      $f_{i,S,T}(M)=0$ for all $S \subset \{1,\hdots,K\}$, $T \subset \{1,\hdots,n_i\}$
       with $|S|=|T|=k_i+1$, for $i=1,\hdots,\ell-1$, where
       $f_{i,S,T}(M)$ is the 
             minor of $M_i$ corresponding to the rows and columns indexed by $S$ and $T$ respectively.
   \end{remark}

\subsection{Convolutional Codes}
\begin{definition}
    For $k\leq n$, an $(n,k)$ \textbf{convolutional code} $\mathcal{C}$ is defined as $\mathbb{F}_{q}[z]$-submodule of $\mathbb{F}_{q}[z]^n$ of rank $k$. A polynomial matrix $G(z)\in\mathbb{F}_{q}[z]^{k\times n}$ of full (row) rank such that
    $$
\mathcal{C} =  \{v(z) = u(z)G(z) \in \mathbb{F}_{q}[z]^n\ |\ u(z) \in \mathbb{F}_{q}[z]^k\} \nonumber
$$
is called \textbf{generator matrix} of $\mathcal{C}$.
The maximum (polynomial) degree of the full-size minors of a generator matrix is called the \textbf{degree} $\delta$ of $\mathcal{C}$. An $(n,k)$ convolutional code of degree $\delta$ is denoted as $(n, k, \delta)$ convolutional code.\\
If there exists a matrix $H(z)\in\mathbb F_q[z]^{(n-k)\times n}$ of full (row) rank such that $$\mathcal{C}=\{v(z)\in\mathbb F_q[z]^n\ |\ H(z)v(z)^\top=0\},$$ this matrix $H(z)$ is called \textbf{parity-check} matrix for $\mathcal{C}$.
\end{definition}

In contrast to linear block codes not all convolutional codes admit a representation via a parity-check but only those with a generator matrix $G(z)\in\mathbb F_q[z]^{k\times n}$ that is left-prime, i.e. $G(z)$ is full rank for all $z\in\overline{\mathbb F}_q$, where $\overline{\mathbb F}_q$ denotes the algebraic closure of $\mathbb{F}_q$; see \cite{York97}.
Such convolutional codes are called \textbf{non-catastrophic}.

The structure of convolutional codes makes them very suitable for low-delay applications and the crucial distance measure for such applications are the column distances.

\begin{definition}
For $j\in\mathbb N_0$, the \textbf{$j$-th column distance}\index{column distance} of a convolutional code $\mathcal{C}$ is 
$$
d_j^c(\mathcal{C}):=\min\left\{\sum_{t=0}^j wt(v_t)\ |\ {v}(z)=\sum_{i\in\mathbb N_0}v_iz^i\in\mathcal{C} \text{ with }{v}_0 \neq 0\right\}.
$$
\end{definition}

For the $j$-th column distance, one has the following upper bound.

\begin{theorem}\cite{MDS-Conv}\label{thm:column} Let $\mathcal{C}$ be an $(n,k,\delta)$ convolutional code. Then, for $j\in\mathbb N_0$,
$$d_j^c (\mathcal{C}) \leq (n-k)(j + 1) + 1.$$
If $d^c_j = (n-k)(j+1)+1$ for some $j \in \mathbb N_0$, then $d^c_i = (n-k)(i+1)+1$, for $i \leq j$. 
Moreover, $d_j^c (\mathcal{C}) = (n-k)(j + 1) + 1$ implies  $j\leq L
:=\left\lfloor\frac{\delta}{k}\right\rfloor+\left\lfloor\frac{\delta}{n-k}\right\rfloor.$
\end{theorem}



\begin{definition}\label{def:MDP_code}
An $(n,k,\delta)$ convolutional code $\mathcal{C}$ is called \textbf{maximum distance profile (MDP)} if $d_L^c (\mathcal{C}) = (n-k)(L + 1) + 1$.
\end{definition}
Let $G(z)=\sum_{i=0}^{\mu}G_iz^i$ with $G_{\mu}\neq 0$, where $\mu:=\deg(G(z))$ is called \textbf{memory} of $G(z)$.
A convolutional code that possesses a generator matrix $G(z)$ with $\mu=1$ is called \textbf{unit-memory} if $rk(G_1)=k$ and \textbf{partial unit-memory} if $rk(G_1)<k$. 
  The \textbf{sliding generator matrix} 
 is defined as $$G_j^c:=\begin{bmatrix}
G_0 & \cdots & G_{j}\\
 & \ddots & \vdots\\
\mathbf{0} &  & G_0
\end{bmatrix}\in\mathbb F_q^{(j+1)k \times (j+1)n}\quad
\text{for}\quad j\in\mathbb N_0,
$$ 
with $G_i=0$ for $i>\mu$. Similarly, if $\mathcal{C}$ is non-catastrophic with parity-check matrix $H(z)=\sum_{i=0}^\nu H_iz^i$, then the 
\textbf{sliding parity-check matrix} 
 is defined as $$H_j^c:=\begin{bmatrix}
H_0 &  & \mathbf{0} \\
\vdots  & \ddots & \\
H_{j}& \hdots  & H_0
\end{bmatrix}\in\mathbb F_q^{(j+1)(n-k) \times (j+1)n}\quad
\text{for}\quad j\in\mathbb N_0.
$$ 
Next, we present a criterion to check whether a convolutional code is MDP.

\begin{theorem}\cite{MDS-Conv}\label{thm:MDP_criterion}
Let $\mathcal{C}$ be a convolutional code with generator matrix $G(z)=\sum_{i=0}^{\mu}G_iz^i\in \mathbb{F}_q[z]^{k\times n}$ and $0\leq j\leq L$.
The following statements are equivalent:
  \begin{itemize}
  \item[(a)] $d_j^c (\mathcal{C})=(n-k)(j + 1) + 1$.
  \item[(b)]
  All non-trivial full-size minors of  $G_j^c\in\mathbb F_q^{(j+1)k \times (j+1)n}$ are nonzero, i.e. all minors
  formed
by columns with indices $ 1 \leq t_1 < \cdots < t_{(j+1)k}$, where $t_{sk+1} > sn$, for $s=1,2, \ldots, j$.
\end{itemize}
In case $\mathcal{C}$ is non-catastrophic with left-prime parity-check matrix $H(z)=\sum_{i=0}^\nu H_iz^i$, then (a) and (b) are further equivalent to the following statement:
\begin{itemize}
 \item[(c)]
 All non-trivial full-size minors of $H_j^c\in\mathbb F_q^{(j+1)(n-k) \times (j+1)n}$ are nonzero, i.e. all minors formed by columns with indices
    $1\leq \ell_1<\cdots<\ell_{(j+1)(n-k)}\leq(j+1)n$ which fulfill
    $\ell_{s(n-k)}\leq sn$ for $s=1,\hdots,j$.
  \end{itemize}
\end{theorem}

Note that condition $(b)$ of the previous theorem implies that $G_0$ has to be an MDS matrix (this is due to the case where we choose $k$ out of the first $n$ columns of $G_j^c$). Moreover, the previous theorem implies the following statement about the dual of an MDP code.

\begin{corollary}\label{dual}\cite{MDS-Conv}
Let $G(z)$ be a left-prime generator matrix of an $(n,k,\delta)$ MDP code. Then $G(z)$ is the parity-check matrix of an $(n,n-k,\delta)$ MDP code. In particular, the dual of an $(n,k,\delta)$ MDP code is an $(n,n-k,\delta)$ MDP code.
\end{corollary}

In this paper, we only consider partial unit-memory convolutional codes with $k>n-k=\delta$, i.e. $L=1$; note that this class of codes is also considered in \cite{luo2023construction,chen2023lower}.

\begin{remark}\label{mdp_matrixcompletion}
The construction of MDP codes can be stated as a matrix completion problem for $M=G^c_L$. For $L=1$, the constraints are given by $m_{i,j}=0$ for $(i,j)\in\{k+1,\hdots,n\}\times\{1,\hdots,n\}$ and $m_{i,j}-m_{i+k,j+n}=0$ for $(i,j)\in\{1,\hdots,k\}\times\{1,\hdots,n\}$.
\end{remark}

\section{Construction of MR-LRCs with UL}


In the following, we present a construction for MR-LRCs with UL with very sparse generator matrices.
To this end, we consider a generator matrix of the form
    \begin{equation}\label{G}
G:=\begin{bmatrix}
    G_0 & 0 & 0 & \dots & 0  & P_0 \\ 
     0 & G_1 & 0 & \dots & 0  & P_1 \\ 
    \vdots &\vdots  & \vdots  & \ddots &  \vdots  & \vdots \\ 
     0 & 0 & 0 &  \dots & G_{\ell-1}  & P_{\ell-1} \\ 
\end{bmatrix}
\end{equation}
with $G_i\in\mathbb F^{k_i \times n_i}$ and $P_i\in\mathbb F^{k_i\times h}$ for $i=0,\hdots,\ell-1$. 
Moreover, if a message $ m \in \mathbb{F}^K$ is partitioned as $m = (m_0,...,m_{\ell -1})$ with \textbf{information blocks} $m_i\in\mathbb F^{k_i}$
for $i=0,\hdots,\ell-1$, then
MR-LRCs with UL of this form have the property that
 $m_i$ can be recovered from $m_i G_i$ if not more than $n_i-k_i$ erasures happened in the coordinates with indices in $\mathcal{G}_i$. This property generalizes the notion of information symbol locality (see e.g. \cite{pkglk12}) to what we call
\textbf{information block locality}.
In this way we further specify the matrix completion problem given in Remark \ref{compmr} via considering generator matrices of the form \eqref{G} where we additionally fix $G_i$ to be Cauchy matrices. In this setup, solving the corresponding matrix completion problem asks, given the structure of $G$ and Cauchy matrices $G_i$, to find $P_i$ such that all non-trivial full-size minors of $G$ are nonzero.

 \begin{theorem}\label{thm_MRC}
Let $h\leq\min\{k_{i},\ 0\leq i\leq \ell-1 \}$ and let $G_i\in\mathbb F_q^{k_i\times n_i}$, for $i=0,\hdots,\ell-1$, be Cauchy matrices and  \begin{equation}
   P_i:= \begin{bmatrix}   
 \mathrm{Diag}(\alpha^{0}, \alpha^{i}, \alpha^{2i},...,\alpha^{(h-1)i}) \\
 0_{(k_i-h)\times h}
   \end{bmatrix}
\end{equation}  
for some primitive element $\alpha\in\mathbb F_{q^d}$. If $h=1$ and $d\geq 1$, or if $h>1$ and
$$d\geq(\ell-1)\cdot\left\lfloor \frac{h}{2}\right\rfloor\cdot \left\lceil \frac{h}{2}\right\rceil+1=:D,$$ then the corresponding code is MR-LRC with UL over $\mathbb F_{q^d}$.
 \end{theorem}

\begin{proof}
Recall that we need to show that if we puncture the $i$-th local repair group $\mathcal{G}_i$ in $n_i-k_i$ positions, then the punctured code is MDS. Therefore, we can assume that the $G_i$ are square Cauchy matrices and need to show that the corresponding code is MDS.
All fullsize minors of $G$ are polynomials over $\mathbb F_q$ in $\alpha$ and we will show that all non-trivial full-size minors are  nonzero polynomials whose difference between the degrees of the maximal and minimal nonzero term is at most $D-1$.

Consider first the case that there are no erasures in the last $h$ columns of $G$, i.e. no erasures in the global parity symbols.

For $i=0,\hdots,\ell-1$, denote by $x_i$ the number of erasures in the columns corresponding to $G_i$ (after the puncturing that causes that the $G_i$ are square matrices), 
i.e.
$x_0+\dots+x_{\ell-1} = h.$
When calculating the corresponding minor, each nonzero term in the determinant corresponds to choosing 
 $x_i$ entries in $P_i$, for all $i=0,\hdots,\ell-1$, where all the choices (over all $P_i$) are from distinct columns (out of the last $h$ columns) of $G$.
 Let $w_{i,j}$ be the $j^\mathrm{th}$ choice for $P_i$ and let $T_{\vec{x}}$ be set of all such choices.   
Then, the minor has the form 
\begin{equation}
     \sum_{w \in T_{\vec{x}}} M_{w} \prod_{i=0}^{\ell - 1}  \alpha^{iw_{i,0}} \cdot \alpha^{iw_{i,1}} \cdot ... \cdot \alpha^{iw_{i,x_i-1}}=
     \sum_{w \in T_{\vec{x}}} M_{w}   \alpha^{ \sum_{i=0}^{\ell - 1} i \sum_{j \in [x_i]} w_{i,j}}
\end{equation}
where $M_w$ is a product of minors of Cauchy matrices and hence, nonzero.
We obtain the minimal exponent in $\alpha$ when choosing 
$    w_{\ell-1,j} =  j$, $
    w_{\ell - 2,j} = x_{\ell-1} + j, \cdots, 
    w_{0,j} = x_1 + x_2 +\cdots +x_{\ell-1} + j$.
On the other hand, we obtain the maximal exponent in $\alpha$ for
$    w_{0,j} = j$, $
    w_{1,j} = x_0 +j$, $\cdots$,
      $w_{\ell - 1,j} = x_0 + x_1 +\cdots +x_{\ell-2} + j$.
Therefore, the difference between the degrees of the maximal and minimal nonzero term is equal to
\begin{align*}
   & \sum_{i=0}^{\ell - 1} i \sum_{j \in [x_i]}w_{i,j}^\mathrm{max} - w_{i,j}^\mathrm{min}\\
   &
   =
    \sum_{i=0}^{\ell - 1} i \sum_{j \in [x_i]} x_0+\cdots +x_{i-1}+j -x_{i+1}-\cdots-x_{\ell-1 }-j\\
     &=     \sum_{i=0}^{\ell - 1}i x_i (x_0+\cdots+x_{i-1} -x_{i+1}-\cdots-x_{\ell-1 })\\
     &
   = \sum_{0 \leq i<j\leq \ell-1} (j-i) x_ix_j=: Q_h(x_0,\ldots,x_{\ell-1})
    \end{align*}
Hence, we have to maximize the quadratic form 
$
  Q_h(x_0,\ldots,x_{\ell-1})
$
subject to the constraints $x_0+\cdots+x_{\ell-1}=h$ and $x_0,\hdots,x_{\ell-1}\in\mathbb N_0$.
Note that $Q_1(x_0,\hdots,x_{\ell-1})=0$, which shows the claim for the case that $h=1$.
In the following, we assume $h>1$ and we will show that for any $(x_0,\hdots,x_{\ell-1})\in\mathbb N_0^\ell$ with $x_0+\cdots+x_{\ell-1}=h$, there exists $(\hat{x}_0,\hdots,\hat{x}_{\ell-1})\in\mathbb N_0^\ell$ with $\hat{x}_0+\cdots+\hat{x}_{\ell-1}=h$ and $\hat{x}_1=\cdots=\hat{x}_{\ell-2}=0$ such that 
$Q_h(x_0,\ldots,x_{\ell-1})\leq Q_h(\hat{x}_0,\hdots,\hat{x}_{\ell-1})$.

If $x_1=\cdots=x_{\ell-2}=0$, we are done. Otherwise, choose $t\in\{1,\dots,\ell-2\}$ such that $x_{t}\neq 0$. We want to shift the weight of $x_t$ to the edges of the sequence without decreasing the value of $Q_h$. If 
\begin{equation}\label{criterion}
\sum_{r=0}^{t-1}x_r\geq \sum_{r=t+1}^{\ell-1}x_r,
\end{equation}
we shift the weight to the right hand side, i.e. we
do the transformations $x_t\mapsto x_t-1$ and $x_{t+1}\mapsto x_{t+1}+1$. If \eqref{criterion} is not true, then we shift the weight to the left, i.e. we do the transformations $x_t\mapsto x_t-1$ and $x_{t-1}\mapsto x_{t-1}+1$.
If \eqref{criterion} is true, then 
\begin{align*}
&Q_{h}(x_0,\hdots,x_t-1,x_{t+1}+1,\hdots,x_{\ell-1})-Q_h(x_0,\hdots,x_{\ell-1})\nonumber\\
&= 
\sum_{r=0}^{t-1} (t+1-r)\,x_r + \sum_{r=t+2}^{\ell-1} (r-t-1)\,x_r-\sum_{r=0}^{t-1} (t-r)\,x_r - \sum_{r=t+2}^{\ell-1} (r-t)\,x_r+\nonumber\\
&+(x_t-1)(x_{t+1}+1)-x_tx_{t+1}=\sum_{r=0}^tx_r-\sum_{r=t+1}^{\ell-1}x_r-1\geq x_t-1\geq 0.    
\end{align*}
Analogous, if \eqref{criterion} is not true, then
\begin{align*}
&Q_{h}(x_0,\hdots,x_{t-1}+1,x_t-1,\hdots,x_{\ell-1})-Q_h(x_0,\hdots,x_{\ell-1})
\nonumber\\
&= 
\sum_{j=0}^{t-2} (t-1-j)\,x_j \;+\; \sum_{j=t+1}^{\ell-1} (j-t+1)\,x_j-\sum_{j=0}^{t-2} (t-j)\,x_j \;-\; \sum_{j=t+1}^{\ell-1} (j-t)\,x_j\nonumber\\
&+(x_t-1)(x_{t-1}+1)-x_tx_{t-1}\nonumber\\
&=-\sum_{j=0}^{t-1}x_t+\sum_{j=t}^{\ell-1}x_t-1> x_t-1\geq 0.    
\end{align*}
Note that condition \eqref{criterion} ensures that if some weight is moved to the right in some step, it can never be moved back towards the left in a later step, and the other way round.
In summary, this shows that we can shift all the weight to the edges $x_0$ and $x_{\ell-1}$ in some way without decreasing $Q_h$.  
Hence, 
$Q_h$ attains this maximum for $x_0=\left\lfloor \frac{h}{2}\right\rfloor$, $x_1=\cdots=x_{\ell-2}=0$, $x_{\ell-1}=
\left\lceil \frac{h}{2}\right\rceil$ and this maximum is equal to
$D=(\ell-1)\cdot \left\lfloor \frac{h}{2}\right\rfloor\cdot \left\lceil \frac{h}{2}\right\rceil$ 
. 
Note that if $h$ is odd, this maximum is attained for several choices of $x_0,\hdots,x_{\ell-1}$ but we are only interested in the maximal value.

It is easy to see that the difference between maximal and minimal degree in $\alpha$ of the nonzero terms of any non-trivial full-size minor of $G$ is smaller than $D-1$ if we have erasures in the last $h$ columns of $G$. 
\end{proof}



Our construction has the advantage of a very sparse generator matrix leading to low encoding complexity and good update efficiency as explained in the following.
To analyze the encoding complexity for the codes of the previous theorem, we follow the standard convention to only count the number of multiplications, but in our case, the number of additions is comparable. 
To calculate $mG$ for $m\in\mathbb F_{q^d}^K$ and $G$ as in \eqref{G}, for each $i=0,\hdots,\ell-1$, we need to multiply a vector from $\mathbb F_{q^d}^{k_i}$ with an $k_i\times n_i$ Cauchy matrix with entries in $\mathbb F_q$, which requires $O\left(\sum_{i=0}^{\ell-1}n_i\log^2 n_i\right)$ multiplications in $\mathbb F_{q^d}$ (see \cite{doi:10.1137/0908017}), and additionally we need to multiply $k_0+\cdots+k_{\ell-1}$ elements of the form $\alpha^t$ for some $t\in\mathbb N_0$ with an element of $\mathbb F_{q^d}$. Hence, we get the following result: 

\begin{theorem}
 The encoding complexity for our construction given in Theorem $\ref{thm_MRC}$ is $O\left(\sum_{i=0}^{\ell-1}n_i\log^2 n_i\right)$.
\end{theorem}
Moreover, if some message symbol needs to be updated, i.e., its value to be changed, then not all of the codeword symbols need to be recomputed; one only has to update the corresponding local group and exactly one global parity symbol, since the diagonal structure of the matrices $P_0,...,P_{\ell-1}$ implies that each message symbol is only involved in the calculation of one global parity symbol.

\section{Construction of partial unit-memory MDP convolutional codes}

In this section, we present constructions for partial unit-memory MDP codes.

\begin{lemma}\label{noncat}
Let $G(z)=G_0+G_1z\in\mathbb F_q[z]^{k\times n}$, where $G_0$ is an MDS matrix and $G_1$ is of the form $G_1 = [ X \ \mathbf{0}_{k \times k}]$. Then, 
the corresponding convolutional code is non-catastrophic.
\end{lemma}
    
\begin{proof}
For all $\hat{z}\in\overline{\mathbb F}_q$, the last $k$ columns of $G(\hat{z})$ are equal to the last $k$ columns of $G_0$ and hence, $G(\hat{z})$ is full rank for all $\hat{z}\in\overline{\mathbb F}_q$.
\end{proof}
In the construction for partial unit-memory MDP codes from \cite{itw},
\begin{equation}\label{eqn_Matrix_X}
    X:=\begin{bmatrix}   
 \mathrm{Diag}(\alpha, \alpha^{2}, \alpha^{3},...,\alpha^{n-k}) \\
 0_{(2k-n)\times (n-k)}
   \end{bmatrix}
 %
  \in \mathbb{F}_{q^d}^{k \times (n - k)},
\end{equation}
\normalsize
where $\alpha \in \mathbb{F}_{q^d}$ has minimal polynomial over $\mathbb F_q$ of degree $d$.

\begin{theorem}\cite{itw}\label{diag}
Let $k>n-k=\delta$ and $d = \left \lceil\frac{\delta^2-1}{4}\right \rceil + 1$. Let $G_0 \in \mathbb F_{q}^{k\times n}$ be a superregular matrix and $G_1 = [X \ \mathbf{0}_{k \times k}] \in \mathbb{F}_{q^d}^{k \times n}$ with $X$ given by equation \eqref{eqn_Matrix_X}. Then, all nontrivial minors of $G_1^c$
are nonzero, i.e., 
we get an $(n,k,\delta)$-MDP code over $\mathbb F_{q^d}$.
\end{theorem}

The previous theorem can be seen as special case of the matrix completion problem from Remark \ref{mdp_matrixcompletion} in the sense that we additionally fix $G_0$ to be {\bf superregular} (i.e., a matrix whose every minor is nonzero) and $G_1$ to be of the form 
$[X \ \mathbf{0}_{k \times k}]$ for some matrix $X$ to be determined.
In the following, we present alternative constructions for partial unit-memory MDP codes using a Vandermonde matrix for $G_0$ and a different structure for the matrix $X$. 
These constructions have the advantage that they require smaller finite fields but the disadvantage that we can only cover the special cases $\delta=2$ and $\delta=3$. In particular, the field size for our previous 
construction,
\cite{itw}, is $q = O((n+k)^\delta)$, whereas the following construction has $q=O(n^\delta)$.
For $\beta_1, \dots, \beta_{n-k}\in\mathbb F_{q^{n-k}}$, we define
\vspace{-2mm}
\begin{equation}\label{xhat}
    \hat{X} := \begin{bmatrix}
         0_{(2k-n)\times (n-k)}\\
        T
    \end{bmatrix}\quad\text{with}\quad 
    T:=\begin{bmatrix}
            \beta_1     & 0 &     \cdots     &     0    \\
        \beta_2       & \ddots &   \ddots      &    \vdots      \\
        \vdots        & \ddots &     \ddots     & 0       \\
        \beta_{n - k} & \cdots & \beta_2 & \beta_1 
    \end{bmatrix}.
\end{equation}

\begin{theorem}\label{vander1}
    Let $q\geq n \geq 4$, $k = n - 2$. Let $G_0\in\mathbb F_q^{k\times n}$ be a Vandermonde matrix and $\hat{X}$ a $k \times 2$ matrix as in \eqref{xhat} where $\beta_1, \beta_2 \in \mathbb F_{q^2}$ are linearly independent over $\mathbb F_q$. Then, all nontrivial minors of
    $G_1^c$
    are nonzero, i.e., we get an $(n,k,2)$ MDP code over $\mathbb F_{q^2}$.
\end{theorem}

\begin{proof}
    Choosing integers $i, j, \ell$, we consider any $2k \times 2k$ submatrix $M$ of $G_1^c$ by choosing $i$ columns out of the first $n$ columns, $j$ columns out of the next $n-k$ and $\ell$ columns out of the last $k$ columns. If $M$ corresponds to a non-trivial minor, then
$$
        i +j +\ell = 2k,\quad \text{ where } \quad 0 \leq i \leq k, \  0 \leq j \leq n-k=2, \ 0\leq\ell\leq k.
$$
Then, we can write $M$ as $2k \times 2k$ block matrix of the form
$$
    M= \begin{bmatrix}
        U & V \\
        \mathbf{0} & W\\    
    \end{bmatrix}
    \quad\text{with}\quad V=[\hat{V}\ 0_{k\times\ell}], 
$$
    where $U$ is a submatrix of $G_0$ of size $k \times i$, $V$ is a submatrix of $G_1$ of size $k \times (j+\ell)$, and $W$ is a submatrix of $G_0$ of size $k \times (j+\ell)$. We need to show that $\det(M) \neq 0$. Since $n-k=2$, there exist only 3 choices for $i$, namely $i\in\{k, k-1, k-2\}$ and these are considered as follows:   
    \begin{itemize}
        \item[(i)] If $i=k$, then   $(i,j, \ell) = (k,j, k-j)$ for any $j\in\{0,1,2\}$. Since $U$ and $W$ are $k\times k$ submatrices of the MDS matrix $G_0$, one obtains that $\det(M) \neq 0$.
        
            
        \item[(ii)] If $i=k-1$, then  $(i,j,\ell) \in \{(k-1,1,k),(k-1,2,k-1)\}$.
        If we consider the subcase $(k-1,1,k)$, then $\det(M)=\det([U\ \hat{V}])\cdot\det(W)$ where $\det(W)\neq 0$ since it is a $k \times k$ minor of the MDS matrix $G_0$. Now, to see  $\det([U\ \hat{V}])\neq 0$, we choose exactly $1$ column from $X$. If we choose the first column, we obtain $\det([U\ \hat{V}])= a_1 \beta_1 + a_2\beta_2$ with $a_1,a_2\in\mathbb F_q$ being $(k-1)\times (k-1)$ minors of $G_0$. Since $G_0$ is Vandermonde, we obtain that $a_i\neq 0$ for (at least) one $i\in\{1,2\}$ and since $\beta_1, \beta_2$ are linearly independent over $\mathbb F_q$, we obtain that $\det(M) \neq 0$. On the other hand, if we choose the second column of $X$, then $\det([U\ \hat{V}]) = a_1 \beta_1$, where $a_1$ is the determinant of a $(k-1) \times (k-1)$ Vandermonde matrix obtained by deleting the last row of $U$ and hence $a_1 \neq 0$.
        Therefore, $\det(M) \neq 0$. 
      
        Now, in the subcase $ (k-1,2, k-1)$,  one has $M=\begin{bmatrix}
    A  & B\\ C & D
    \end{bmatrix}$ with $A=\begin{bmatrix}
    \tilde{g}_1 & \cdots& \tilde{g}_{k-1} & x_{1} \end{bmatrix}$, where the $\tilde{g}_i$ are different columns from $G_0$ and $x_{1}$ is the first column of $X$, $B=\begin{bmatrix} x_{2} & 0 & \cdots & 0\end{bmatrix}$ with $x_{2}$ being the second column of $X$, $C=\begin{bmatrix}
        0, & \cdots & 0 & g_{1} \end{bmatrix}$ with $g_{1}$ being the first column of $G_0$ and $D=\begin{bmatrix}
                g_{2}& g_{i_3}& \cdots & g_{i_{k+1}}
            \end{bmatrix}$ with $g_{i_j}$ the $i_j$-th column of $G_0$ and $3 \leq i_3 <\hdots <i_{k+1} \leq n$. Note that $\det(D) \neq 0$ and thus, $\det(M) \neq 0$ if and only if 
 $\det(A-BD^{-1}C)$ is nonzero. Moreover, 
 $$
A-BD^{-1}C=\begin{bmatrix}
    \tilde{g}_1& \hdots & \tilde{g}_{k-1}&  \begin{bmatrix}
    0\\ \vdots\\ 0\\ \beta_1\\ \beta_2\\ 
    \end{bmatrix}-\begin{bmatrix}
    0\\ \vdots\\ 0\\ 0\\ \beta_1 \langle\tilde{d_1}, g_{1}\rangle
    \end{bmatrix} 
\end{bmatrix}, 
$$
    where $\tilde{d_1}$ is the first row of $D^{-1}$. Hence, $\det(A-BD^{-1}C)=a_1 \beta_{1}+ a_2 \beta_2$ with $a_2\in\mathbb F_q^\ast$, being a $(k-1) \times (k-1)$ minor of the Vandermonde matrix $G_0$. Since, $\beta_1$ and $\beta_2$ are independent over $\mathbb F_q$, we obtain $\det(M)\neq 0$. 

    \item[(iii)] If $i=k-2$, then the only possibility for $(i,j,\ell)$ is $(k-2, 2, k)$. Then, $\det(M)=\det([U\ \hat{V}])\cdot\det(W)$ where $\det(W)\neq 0$ since it is a $k \times k$ minor of the MDS matrix $G_0$. One obtains that $\det([U\ \hat{V}])= a\beta_1^2$ where $a$ is the determinant of a $(k-2)\times (k-2)$ Vandermonde matrix obtained by deleting the last two rows of $U$, i.e. $a\in\mathbb F_q^\ast$. Thus, $\det(M) \neq 0$. 
    \end{itemize}
Thus, $\det(M) \neq 0$ for all choices and we get $(n,k,2)$-MDP code over $\mathbb F_{q^2}$. 
\end{proof}

\begin{theorem}\label{vander2}
    Let $q\geq n$, $k=n-3$ and $G_0\in\mathbb F_q^{k\times n}$ be a Vandermonde matrix. Let $G_1=[\hat{X}\ 0_{k\times k}]$ with $\hat{X}$ as in \eqref{xhat} with $\beta_1=z$, $\beta_2=z^2$, $\beta_3=1$ where $z\in\mathbb F_{q^3}\setminus\mathbb F_q$ such that the minimal polynomial of $z$ has constant coefficient unequal to $-1$. 
    Then, all nontrivial minors of
     $G_1^c$
    are nonzero, i.e., we get an $(n,k,3)$ MDP code over $\mathbb F_{q^3}$.
\end{theorem}

\begin{proof}
For integers $i, j, \ell$, we consider any $2k \times 2k$ submatrix $M$ of $G_1^c$ by choosing $i$ columns out of the first $n$ columns, $j$ columns out of the next $n-k$ and $\ell$ columns out of the last $k$ columns. If $M$ corresponds to a non-trivial minor, then
    $i +j +\ell = 2k$ where $1 \leq i \leq k, \  0 \leq j \leq n-k=3$ and $0 \leq \ell \leq k$. Thus, we can write 
$$
 M= \begin{bmatrix}
    U & V \\
    \mathbf{0} & W\\
\end{bmatrix}
\quad\text{with}\quad V=[\hat{V}\ 0_{k\times\ell}], 
$$
where $U$ is a submatrix of $G_0$ of size $k \times i$, $V$ is a submatrix of $G_1$ of size $k \times (j+\ell)$, $W$ is a submatrix of $G_0$ of size $k \times (j+\ell)$. 
For fixed values of $(i,j,\ell)$, we need to show that $\det(M) \neq 0$. We consider different cases, for which we use the abbreviation $(i,j,\ell)$.\\
\underline{Case $(k,j, k-j)$}: For any $j\in\{0,1,2,3\}$,
this case is clear from the fact that $G_0$ is an MDS matrix.

Now, we consider the case when  $\ell=k$. In this case, $$\det(M)=\det([U\ \hat{V}])\cdot\det(W)$$ where $\det(W)\neq 0$ since it is a fullsize minor of the MDS matrix $G_0$. To see that also $\det([U\ \hat{V}])\neq 0$, i.e. to see that $[G_0\ \hat{X}]$ is an MDS matrix, we consider different cases.\\
\underline{Case $(k-1,1,k)$}: One  obtains
$$\det([U\ \hat{V}])=a_0+a_1z+a_2z^2$$ with $a_0,a_1,a_2\in\mathbb F_q$ being $(k-1)\times (k-1)$ minors of $G_0$. Since $G_0$ is Vandermonde, one has  that $a_i\neq 0$ for (at least) one $i\in\{1,2,3\}$,
and since $1,z,z^2$ are linearly independent over $\mathbb F_q$ 
the corresponding minor is nonzero.\\
\underline{Case $(k-3,3,k)$}: One obtains that $\det([U\ \hat{V}])=
az^3$ with
$a\in\mathbb F_q^\ast$.\\
\underline{Case $(k-2,2,k)$}: We have to distinguish three sub-cases corresponding to the three choices for choosing two out of the three columns of $\hat{X}$.

If we choose the last two columns of $\hat{X}$, $\det([U\ \hat{V}])=bz^2$ where $b$ is the determinant of a $(k-2)\times (k-2)$ Vandermonde matrix, i.e., $b\in\mathbb F_q^\ast$.

If we choose the first and the last column of $\hat{X}$, $$\det([U\ \hat{V}])=cz^3+dz^2=z^2(cz+d)$$ where $c$ is the determinant of a $(k-2)\times (k-2)$ Vandermonde matrix, i.e., $c\in\mathbb F_q^\ast$ and hence, $\det([U\ \hat{V}])\neq 0$.

If we choose the first two columns of $\hat{X}$, $$\det([U\ \hat{V}])=a(z^4-z)+bz^2+cz^3$$ 
where $a$ is the determinant of a $(k-2)\times (k-2)$ Vandermonde matrix, i.e.,  $a\in\mathbb F_q^\ast$. Thus, $\det([U\ \hat{V}])\neq 0$ if and only if $p(z)=z^3+\frac{c}{a}z^2+\frac{b}{a}z-1\neq 0$, which is true because the minimal polynomial of $z$ has degree $3$ and cannot be $p(z)$.
\\
\underline{Case $(k-1,2,k-1)$}: $M=\begin{bmatrix}
    A  & B\\ C & D
\end{bmatrix}$
with $A=\begin{bmatrix}
    \tilde{g}_1& \cdots & \tilde{g}_{k-1}& x_{i_1}\end{bmatrix}$, where the $\tilde{g}_i$ are different columns from $G_0$ and $x_{i_1}$ is 
column $i_1$ of $\hat{X}$,
$B=\begin{bmatrix}
    x_{i_2} & 0& \cdots& 0\end{bmatrix}$ with $x_{i_2}$ being column $i_2$ from $\hat{X}$,
$C=\begin{bmatrix}
    0& \cdots& 0& g_{i_1}\end{bmatrix}$ with $g_{i_1}$ the $i_1$-th column of $G_0$ and $D=\begin{bmatrix}
        g_{i_2}& g_{i_3}& \cdots& g_{i_{n-2}}    \end{bmatrix}$ with $g_{i_j}$ the $i_j$-th column of $G_0$, $1\leq i_1<i_2<\cdots<i_{n-3}<i_{n-2}\leq n$.
In the following, we will use that
$\det\begin{bmatrix}
   A & B\\ C & D
\end{bmatrix}$ is nonzero if and only if $\det(A-BD^{-1}C)$ is nonzero.

If $i_1=2$, i.e. $i_2=3$, then for $$v^{T} = \begin{bmatrix}
    0 & \cdots & 0 & z & z^2 - z \langle\tilde{d},g_{i_1}\rangle
\end{bmatrix}$$ where $\tilde{d}$ is the first row of $D^{-1}$, one has 
$$
A-BD^{-1}C=\begin{bmatrix}
    \tilde{g}_1 & \cdots &  \tilde{g}_{k-1} & v 
\end{bmatrix}.
$$ 
Hence, $\det(A-BD^{-1}C)=az^2+bz$ with
$a\in\mathbb F_q^\ast$. Therefore, $\det(M) \neq 0$. 

Now, if $i_1=1$, $i_2=3$, for $$v_1^{T} = \begin{bmatrix}
    \ 0  & \cdots & 0 &  z & z^2 &  1-z\langle\tilde{d}, g_{i_1}\rangle\ 
\end{bmatrix},$$ we get
$A-BD^{-1}C=\begin{bmatrix}
    \tilde{g}_1 & \hdots & \tilde{g}_{k-1} & v_1
\end{bmatrix}.$ Hence, 
$$\det(A-BD^{-1}C)=az^2+bz+c
$$ with
$c\in\mathbb F_q^\ast$. 
Further, if $i_1=1$, $i_2=2$, then for 
$$v_2^{T} = \begin{bmatrix}
    \ 0   &\cdots & 0 &  z & z^2 -z\langle\tilde{d}, g_{i_1}\rangle &  1-z^2\langle\tilde{d}, g_{i_1}\rangle \ 
\end{bmatrix},$$
we get $A-BD^{-1}C=\begin{bmatrix}   
\tilde{g}_1& \cdots&  \tilde{g}_{k-1}& v_2
\end{bmatrix}.$
Thus, $$\det(A-BD^{-1}C)=az^2+bz+c$$ with
$c\in\mathbb F_q^\ast$, which implies that $\det(M) \neq 0$ for all possible choices of 2 out of 3 columns of $\hat{X}$.\\
\underline{Case $(k-1,3,k-2)$}: We have $M=\begin{bmatrix}
    A  & B\\ C & D
\end{bmatrix}$
with $A=\begin{bmatrix}
    \tilde{g}_1 & \cdots & \tilde{g}_{k-1} & x_{1}\end{bmatrix}$, where the $\tilde{g}_i$ are different columns from $G_0$,
$B=\begin{bmatrix}
    x_{2}& x_3&  0& \cdots & 0\end{bmatrix}$,
$C=\begin{bmatrix}
    0& \cdots& 0& g_{1}\end{bmatrix}$
and $D=\begin{bmatrix}
    g_{2}& g_{3}& g_{i_4} &\cdots & g_{i_{n-2}}
\end{bmatrix}$ with $g_{i_j}$ the $i_j$-th column of $G_0$, $4\leq i_4<\cdots<i_{n-2}\leq n$.
%
Let $\tilde{d}_1$ and $\tilde{d}_2$ be the first and second row of $D^{-1}$, respectively. Then, $$v^{T}=\begin{bmatrix}
    \ 0 & \cdots &  0 &  z &  z^2-z\langle\tilde{d}_1, g_1\rangle &  1-z^2\langle\tilde{d}_1, g_1\rangle-z\langle\tilde{d}_2, g_1\rangle\ 
\end{bmatrix}.$$
Then, $A-BD^{-1}C=\begin{bmatrix}
    \tilde{g}_1& \cdots &  \tilde{g}_{k-1}&  v\end{bmatrix}$ 
 and hence, 
 $$\det(A-BD^{-1}C)=az^2+bz+c$$ with
$c\in\mathbb F_q^\ast$. \\
\underline{Case $(k-2,3,k-1)$}: Now, $M=\begin{bmatrix}
    A  & B\\ C & D
\end{bmatrix}$
with $A=\begin{bmatrix}
    \tilde{g}_1& \cdots&  \tilde{g}_{k-2}& x_{1}& x_2\end{bmatrix}$, where the $\tilde{g}_i$ are different columns from $G_0$,
$B=\begin{bmatrix}
    x_3 & 0& \cdots& 0\end{bmatrix}$,
$C=\begin{bmatrix}
    0 & \cdots & 0& g_{1}& g_2\end{bmatrix}$
and $D=\begin{bmatrix}
    g_{3}& g_{i_4}& \cdots& g_{i_{n-2}}\end{bmatrix}$ with $g_{i_j}$ the $i_j$-th column of $G_0$, $4\leq i_4<\cdots<i_{n-2}\leq n$.
%
For 
\begin{align*}
v_1^{T} &= \begin{bmatrix}
     0 &  \cdots & 0 &  z & z^2 &  1-z\langle\tilde{d}_1, g_1\rangle 
\end{bmatrix},\\
v_2^{T} &= \begin{bmatrix}
    0 &  \cdots &  0 &  0 &  z &  z^2-z\langle\tilde{d}_1, g_2\rangle
\end{bmatrix},
\end{align*}
$A-BD^{-1}C=\begin{bmatrix}
    \tilde{g}_1 & \cdots & \tilde{g}_{k-2} &  v_1 & v_2\end{bmatrix}$ and $$\det(A-BD^{-1}C)=a(z^4-z)+bz^2+cz^3$$ with $a\in\mathbb F_q^\ast$.
This is nonzero if and only if $p(z)=z^3+\frac{c}{a}z^2+\frac{b}{a}z-1\neq 0$, which is true since the minimal polynomial of $z$ has degree $3$ and is not equal to $p(z)$.

Thus, for all possible choices of $(i,j,\ell)$, we have obtained $\det(M) \neq 0$, as desired and hence, we get an MDP code.
\end{proof}

Combining the previous theorems and Corollary \ref{dual} and Lemma \ref{noncat} we obtain the following corollary.

\begin{corollary}
Let $k\in\{2,3\}$ and let $H_0\in\mathbb F_q^{(n-k)\times n}$ be a Vandermonde matrix. Let $H_1=[\hat{X}\ 0]$, where $\hat{X}$ is as in Theorem \ref{vander1} for $k=2$ and $\hat{X}$ is as in Theorem \ref{vander2} for $k=3$.
Then, $H(z)=H_0+H_1z$ is the parity-check matrix of an $(n,k,\delta)$ MDP code with $k=\delta$.
\end{corollary}

\begin{remark}
MDP codes with the same parameters as in the previous corollary, i.e. $(n,k,\delta)$ with $k=\delta\in\{2,3\}$, have recently also been constructed in \cite{cheng2026new}. The required field size for this construction in \cite{cheng2026new} is $q=O(n^{\delta-1})$, which is asymptotically better than $q=O(n^\delta)$ what was the state of the art before this paper and what is also the asymptotically required field size for our construction. However, the actual minimal field size in \cite{cheng2026new} is the smallest prime $p$ with $p\geq 3n-2$ for $k=2$ and the smallest prime $p$ with $p\geq 307 n^2$ for $k=3$. In contrast for our construction it is $q^2$ where $q$ is the smallest prime power with $q\geq n$ for $k=2$ and $q^3$ for $k=3$. Hence, for $k=3$, except for very small rates our construction is MDP over smaller fields than the construction from \cite{cheng2026new}.
\end{remark}


\begin{theorem}\cite{itw}
The encoding of the MDP codes from Theorem \ref{diag}, Theorem \ref{vander1} and Theorem \ref{vander2} requires $O(n\log^2(n))$ multiplications.
\end{theorem}
 \begin{proof}
 For the codes of Theorem \ref{diag} this has been shown in \cite{itw}  and for the other codes it can be shown in a similar way.
  \end{proof}

\section{Conclusion}
We provided constructions for MR-LRCs with UL and MDP codes with similar building blocks, i.e., MDS matrices over some base field and diagonal/lower-triangular matrices over an extension field, using similar proof strategies.
However, even if the two corresponding matrix completion problems 
both require to construct a matrix with a given zero pattern such that all non-trivial full-size minors of this matrix are nonzero, the differences in the constraints 
cause 
the building blocks to be arranged in a different way. 
Thus, 
independent proofs are required and we do not have a direct way of obtaining an MDP code from  MR-LRCs with UL and the other way round. 
Hence, the results of this paper can be seen as a first step towards reaching the goal mentioned in the Introduction, i.e. to develop principled matrix completion frameworks that apply uniformly across a wide range of constraint patterns. As a next step, we aim to develop constructions for MDP codes with $L>1$ and MR-LRCs with UL without the restriction on the number of global parities.


For future research we also want to investigate to which extent MDP codes can be seen as MR-LRCs with UL with hierarchical locality \cite{NL19, RB25}.
Moreover, in \cite{umbertodiego} partial MDP codes were considered, 
where the imposed locality constraints prohibit the codes to be MDP, i.e.\ the column distances can only reach a modified upper bound adapted to the locality setting, and codes reaching these bounds are called partial MDP.
Hence, it is a natural question to investigate whether our techniques can also be used for the construction of partial MDP codes.



\bibliographystyle{alpha}
\bibliography{wcc.bib} 

\end{document}